%% file: constant_time_tester.tex
\documentclass[a4paper,UKenglish,cleveref, autoref, thm-restate]{lipics-v2019}
\hideLIPIcs
\usepackage{tikz}
\newtheorem*{proviso}{Proviso}

\title{Faster Property Testers in a Variation of the Bounded Degree Model}

\input{macros}

\author{Isolde Adler}{University of Leeds, School of Computing, Leeds, UK}{i.m.adler@leeds.ac.uk}{https://orcid.org/0000-0002-9667-9841}{}%TODO mandatory, please use full name; only 1 author per \author macro; first two parameters are mandatory, other parameters can be empty. Please provide at least the name of the affiliation and the country. The full address is optional

\author{Polly Fahey}{University of Leeds, School of Computing, Leeds, UK}{mm11pf@leeds.ac.uk}{https://orcid.org/0000-0002-3781-5313}{}

\authorrunning{I. Adler and P. Fahey} %TODO mandatory. First: Use abbreviated first/middle names. Second (only in severe cases): Use first author plus 'et al.'

\Copyright{Isolde Adler and Polly Fahey} %TODO mandatory, please use full first names. LIPIcs license is "CC-BY";  http://creativecommons.org/licenses/by/3.0/

 %TODO mandatory; please add comma-separated list of keywords


\category{} %optional, e.g. invited paper

\relatedversion{} %optional, e.g. full version hosted on arXiv, HAL, or other respository/website
\supplement{}%optional, e.g. related research data, source code, ... hosted on a repository like zenodo, figshare, GitHub, ...

\acknowledgements{}%optional

\nolinenumbers %uncomment to disable line numbering

\AtBeginDocument{%
  \providecommand\BibTeX{{%
    \normalfont B\kern-0.5em{\scshape i\kern-0.25em b}\kern-0.8em\TeX}}}

\EventEditors{Nitin Saxena and Sunil Simon}
\EventNoEds{2}
\EventLongTitle{40th IARCS Annual Conference on Foundations of Software Technology and Theoretical Computer Science (FSTTCS 2020)}
\EventShortTitle{FSTTCS 2020}
\EventAcronym{FSTTCS}
\EventYear{2020}
\EventDate{December 14--18, 2020}
\EventLocation{BITS Pilani, K K Birla Goa Campus, Goa, India (Virtual Conference)}
\EventLogo{}
\SeriesVolume{182}
\ArticleNo{11}

\begin{document}
\maketitle

\begin{abstract}
   Property testing algorithms are highly efficient algorithms, that 
	come with probabilistic accuracy guarantees. For a property $P$,
	the goal is to distinguish inputs that have $P$ from those that
	are \emph{far} from having $P$ with high probability correctly, by
	querying only a small number of local parts of the input. In property testing on graphs, the  \emph{distance} is measured by the number of edge modifications (additions or deletions), that are 
	necessary to transform a graph into one with property $P$.
	Much research has focussed on the \emph{query complexity} of such algorithms,
	i.\,e.\ the  
	number of queries the algorithm makes to the input, 
	but in view of applications, the 
	\emph{running time} of the algorithm is equally relevant. 
	
	In (Adler, Harwath STACS 2018),
	a natural extension of the bounded degree graph model of property testing to relational 
	databases of bounded degree was introduced, and it was shown that on databases of bounded
	degree and bounded tree-width, every property that is expressible in monadic second-order
	logic with counting (CMSO) is testable with constant query complexity and \emph{sublinear} running time.
	It remains open whether this can be improved to constant running time.

	In this paper we introduce a new model, which is based on the bounded degree model, but the
	distance measure allows both edge (tuple) modifications and vertex (element) 
	modifications.
	Our main theorem shows that on databases of bounded
	degree and bounded tree-width, every property that is expressible in CMSO 
	is testable with constant query complexity and \emph{constant} running time in the new model. 
	We also show that every property
	that is testable in the classical model is testable in our model
	with the same query complexity and running time, 
	but the converse is not true.

	We argue that our model is natural and our 
	meta-theorem showing constant-time CMSO testability supports this. 
\end{abstract}

\section{Introduction}
Extracting information from large amounts of data and understanding its global structure can be an immensely challenging and time consuming task. When the input data is huge, many traditionally `efficient' algorithms are no longer practical. The framework of property testing aims at addressing this problem. Property testing algorithms (\emph{testers}, for short) are given oracle access to the inputs, and their goal is to distinguish between inputs which have a given property $\mathbf P$ or are structurally \emph{far} from having $\mathbf P$ with high probability correctly. This can be seen as a relaxation of the classical yes/no decision problem for $\mathbf P$. Testers make these decisions by exploring only a small number of local parts of the input which are randomly chosen. They come with probabilistic guarantees on the quality of the answer. 
Typically, only a constant number of small local parts are explored and the algorithms often run in constant or sublinear time. This speed up in running time, whilst sacrificing some accuracy, can be crucial for dealing with large inputs. In particular it can be useful for a quick exploration of newly obtained data (e.\,g.\ biological networks). Based on the outcome of the exploration, a decision can then be taken whether to use a more time consuming exact algorithm in a second step.

A \emph{property} is simply an isomorphism-closed class of graphs or relational databases.
For example, each Boolean database query $q$ defines a property $\mathbf P_q$, the class of all databases satisfying $q$.
In the bounded degree graph model~\cite{goldreich2002property}, a uniform upper bound $d$ on the degree of the graphs is assumed.
For a small $\epsilon\in (0,1]$, two graphs $\mathcal{G}$ and $\mathcal{H}$, both on $n$ vertices, are $\epsilon$-\emph{close},
if at most $\epsilon dn$ edge modifications (deletions or insertions in $\mathcal{G}$ or $\mathcal{H}$) are necessary to make $\mathcal{G}$ and $\mathcal{H}$ isomorphic. 
%a graph in $\mathcal P$.
If $\mathcal{G}$ and $\mathcal{H}$ are not $\epsilon$-\emph{close}, then they are called $\epsilon$-\emph{far}.
A graph $\mathcal{G}$ is called $\epsilon$-\emph{close} to a property $\mathbf P$, if $\mathcal{G}$ is $\epsilon$-\emph{close} to a member of $\mathbf P$, and $\mathcal{G}$ is $\epsilon$-\emph{far} from $\mathbf P$ otherwise.
The natural generalisation of this model to relational databases of bounded degree (where a database has degree at most $d$ if each element in its domain appears in at most $d$ tuples) was studied in~\cite{adler2018property}, where two 
%database has degree at most $d$ if each element in its domain appears in at most $d$ tuples, and a 
databases $\mathcal D$ and $\mathcal D'$, both with $n$ elements in the domain, are 
$\epsilon$-\emph{close}, 
if at most $\epsilon dn$ tuple modifications (deletions from relations or insertions to relations) 
are necessary to make $\mathcal D$ and $\mathcal D'$ isomorphic, and $\mathcal D$ and $\mathcal D'$ are $\epsilon$-\emph{far} otherwise. We call this model for bounded degree relational databases the \BDRD model.
%Property testing was first introduced in \cite{RubinfeldS96}, with the motivation of Programme Checking. Since then, multiple models have been introduced to deal with different types of inputs, cf.\ e.\,g.\ the book~\cite{goldreich2017introduction}. 

\textbf{Our contributions.} 
In this paper we propose a new model for property testing on bounded degree relational databases,
which we call the \BDRDnew model, with a distance measure that allows both tuple deletions and insertions, and \emph{deletion and insertion of elements of the domain}. On graphs, this translates to edge insertions and deletions, and \emph{vertex insertions and deletions}.
We argue that this yields a natural distance measure. Indeed, take  
any (sufficiently large) graph $\mathcal{G}$, and let $\mathcal{H}$ be obtained from $\mathcal{G}$ by adding an isolated vertex. Then
$\mathcal{G}$ and $\mathcal{H}$ are $\epsilon$-far for every $\epsilon \in (0,1]$ under the classical distance measure, although they only differ in one vertex. In contrast, our distance measure allows for a small number of vertex modifications. 
While comparing graphs on different numbers of vertices by adding isolated
vertices was done implicitly as part of the study the testability of outerplanar graphs~ \cite{babu2016every},
to the best of our knowledge, such a distance measure has not been considered before as part of a model in property testing, which seems surprising to us.

Formally, in the \BDRDnew model, two databases $\mathcal D$ and $\mathcal D'$ are 
$\epsilon$-\emph{close}, if they can be made isomorphic by at most $\epsilon dn$ \emph{modifications}, where a modification is either, (1) removing a tuple from a relation, (2) inserting a tuple to a relation, (3) removing an element from the domain (and, as a consequence, any tuple containing that element is removed), or (4) inserting an element into the domain. Here $n$ is the minimum of the sizes of the domains of $\mathcal D$ and $\mathcal D'$.
In Section~\ref{sec: the model} we give the full details of our model. We note that the \BDRDnew model differs from the \BDRD model only in the choice of the distance measure. 
While we work in the setting of relational databases, we would like to emphasize that our results
%\isolde{all?}\polly{the only result that is not clear is Lemma \ref{lemma: comparing models 2} as the example we give is only for undirected graphs - but we can choose an untestable property on directed graphs - I have added a comment in the proof of Lemma \ref{lemma: comparing models 2}} 
carry over to (undirected and directed) graphs, as these can be seen as special instances of relational 
databases.

%We assume a uniform upper $d$ bound on the degree of the relational databases, where a database has degree at most $d$ if each element in its domain appears in at most $d$ tuples in relations. 
%The testers do not have access to the whole input database but instead have access via an oracle. 
%For a fixed element $a$ of the domain of the input, the tester can query the oracle for the tuples in any relation that contain $a$ and we assume that these queries can be answered in constant time. 

%
%For $\epsilon\in [0,1]$, a database $\mathcal D $ with domain of size $n$ is \emph{$\epsilon$-close} to satisfying $\mathbf P$, if we can make $\mathcal D$ isomorphic to
%a member of $\mathbf P$ by inserting or removing at most $\epsilon dn$ 
%tuples from relations of $\mathcal D$ (i.\,e.\ at most an `$\epsilon$-fraction' of the maximum possible number $dn$ of tuples in relations of $\mathcal D$). If $\mathcal{D}$ is not $\epsilon$-close to $\mathbf P$ then we say it is \emph{$\epsilon$-far} from $\mathbf P$. An \emph{$\epsilon$-tester} for $\mathbf P$ is a probabilistic algorithm, which is given the size $n$ of the domain of the input and has oracle access to the input database, that with probability at least $2/3$ decides correctly whether the input is in $\mathbf P$ or is $\epsilon$-far from $\mathbf P$. We say that a property $\mathbf{P}$ is \emph{testable} if for every $\epsilon\in (0,1]$ there exists an $\epsilon$-tester that only makes a constant number of queries to the oracle.

It is known that in the bounded degree graph model, every minor-closed property is testable \cite{benjamini2010every}, and, more generally, every hyperfinite graph property is 
%(non-uniformly in $n$\polly{correct?}) %\isolde{yes!}
testable \cite{newman2013every} with constant query complexity. 
However, no bound on the running time can be obtained in these general settings. 
Indeed, there exist hyperfinite properties (of edgeless graphs) that are uncomputable.  
In~\cite{adler2018property},
Adler and Harwath ask which conditions guarantee both low query complexity \emph{and} efficient running time. They prove a meta-theorem stating that, on classes of databases (or graphs) of bounded degree and bounded tree-width, every property that can be expressed by a sentence of monadic second-order logic with counting (CMSO) is testable with \emph{constant} query complexity and \emph{polylogarithmic} running time in the \BDRD model. Treating many algorithmic problems simultaneously, this can be seen as an algorithmic \emph{meta-theorem} within the line of research inspired by Courcelle's famous theorem \cite{courcelle1990graph} that states that each property of relational databases which is definable in CMSO is decidable in linear time on relational databases of bounded tree-width. CMSO extends first-order logic (FO) and hence properties expressible in FO (e.g. subgraph/sub-database freeness) are also expressible in CMSO. Other examples of graph properties expressible in CMSO include
bipartiteness, colourability, even-hole-freeness and Hamiltonicity.
Rigidity (i.\,e.\ the absence of a non-trivial automorphism) cannot be expressed in CMSO (cf.~\cite{courcelle2012graph} for
more details).

Our main theorem (Theorem~\ref{thm: CMSO testability}) shows that in the \BDRDnew model, on classes of databases (or graphs) of bounded degree and bounded tree-width, every property that can be expressed by a sentence of monadic second-order logic with counting (CMSO) is testable with \emph{constant} query complexity and \emph{constant} running time.
The question whether constant running time can also be achieved in the \BDRD model
remains open.

We show that the \BDRDnew model is in fact stronger than the \BDRD model: Any property testable in the \BDRD model is also testable in the \BDRDnew model with the same query complexity and running time (Lemma \ref{lemma: comparing models 1}), but there are examples that show that the converse is not true (Lemma \ref{lemma: comparing models 2}).

%We give an example of a property that is testable in the \BDRDnew model but is not testable in the \BDRD model. This property contains all graphs that are either bipartite or have an odd number of vertices. In the \BDRDnew model every graph is at distance at most $1$ from this property, and hence an $\epsilon$-tester\isolde{not explained} can always accept if $1 \leq \epsilon d n$ (where $n$ is the number of vertices in the input graph) and do a full check of the input otherwise. In the \BDRD model owever, if the input has an even number of vertices, the tester must test if the input is bipartite or far from being bipartite, which is known to require a super-constant number of oracle queries \cite{goldreich2002property}. 

%This opens up the interesting question of which other non-testable properties in the \BDRD model become testable in the \BDRDnew model. We do not attempt to answer this specific question in the current paper but we instead move towards pinpointing the conditions required for extremely efficient (i.\,e.~constant) time testers.
In the future, it would be interesting to obtain a characterisation of the properties that are (efficiently) testable in the \BDRDnew model.

\textbf{Our techniques.}
To prove our main theorem, we give a general condition under which properties are testable in constant time in the \BDRDnew model whereas the fastest known testers for such properties in the \BDRD model run in polylogarithmic time. To describe this condition let us first briefly introduce some definitions. A property $\mathbf{P}$ is \emph{hyperfinite} on a class of databases $\mathbf{C}$ if every database in $\mathbf{P}$ can be partitioned into connected components of constant size by removing only a constant fraction of the tuples such that the resulting partitioned database is in $\mathbf{C}$. Let $r \in \mathbb{N}$, given an element $a$ in the domain of a database $\mathcal{D}$ the \emph{$r$-neighbourhood type} of $a$ in $\mathcal{D}$ is the isomorphism type of the sub-database of $\mathcal{D}$ induced by all elements that are at distance at most $r$ from $a$ in the underlying graph of $\mathcal{D}$, expanded by $a$. The \emph{$r$-histogram} of a bounded degree database $\mathcal{D}$, denoted by $\h_r(\mathcal{D})$, is a vector indexed by the $r$-neighbourhood types, where the component corresponding to the $r$-neighbourhood type $\tau$ contains the number of elements in $\mathcal{D}$ that realise $\tau$. The \emph{$r$-neighbourhood distribution} of $\mathcal{D}$ is the vector $\h_r(\mathcal{D})/n$ where $\mathcal{D}$ is on $n$ elements.
We show that for any property $\mathbf{P}$ and input class $\mathbf{C}$, if $\mathbf{P}$ is hyperfinite on $\mathbf{C}$ and the set of $r$-histograms of the databases in $\mathbf{P}$ are semilinear, then $\mathbf{P}$ is testable on $\mathbf{C}$ in constant time (Theorem \ref{thm: constant time tester}). 
As a corollary we then obtain our main theorem, that every property definable by a CMSO sentence is testable on the class of databases with bounded degree and bounded tree-width in constant time (Theorem~\ref{thm: CMSO testability}).

Alon \cite[Proposition 19.10]{lovasz2012large} proved that for every bounded degree graph $\mathcal{G}$ there exists a constant size graph $\mathcal{H}$ that has a similar neighbourhood distribution to $\mathcal{G}$. However, the proof is based on a compactness argument and does not give an explicit upper bound on the size of $\mathcal{H}$. Finding such a bound was suggested by Alon as an open problem \cite{indyk2011open}.
%We prove a result (Theorem \ref{thrm: small dbs}) that can be seen as a strengthening of Alon's theorem.
 We ask under which conditions on a given property $\mathbf P$, for every member of $\mathbf P$ there exists a constant size database with a similar neighbourhood distribution which is also in $\mathbf P$. We show that for any property $\mathbf{P}$ which is hyperfinite on the input class $\mathbf{C}$ and whose $r$-histograms are semilinear, if a database $\mathcal{D}$ is in $\mathbf{P}$ then there exists a constant size database $\mathcal{D'}$ in $\mathbf{P}$ with a similar neighbourhood distribution but this is not true for databases in $\mathbf{C}$ that are far from $\mathbf{P}$. Furthermore, we obtain upper and lower bounds on the size of $\mathcal{D'}$. We can then use this result to construct constant time testers. We first use the algorithm $\operatorname{EstimateFrequencies}_{r,s}$ (given in \cite{newman2013every} and adapted to databases in \cite{adler2018property}) to approximate the neighbourhood distribution of the input database. Then we only have to check if the estimated distribution is close to the neighbourhood distribution of a constant size database in the property.

As a corollary (Corollary~\ref{cor:alon}), we obtain an explicit bound on the size on graphs $\mathcal{H}$ from Alon's
theorem for `semilinear' properties, i.\,e.\ properties, 
where the histogram vectors of the neighbourhood distributions form a semilinear set.

\subparagraph*{Further related work.}
Other than the work already mentioned in \cite{adler2018property} there are only a handful of results on relational databases that utilise models from property testing. Chen and Yoshida \cite{chen2019testability} study a model which is close to the general graph model (cf.\ e.\,g. \cite{alon2008testing}) in which they study the testability of homomorphism inadmissibility. Ben-Moshe et al. \cite{ben2011detecting} study the testability of near-sortedness (a property of relations that states that most tuples are close to their place in some desired order). 
Our model differs from both of these, as it relies on a degree bound and uses different types of oracle access.
Explicit bounds for Alon's theorem restricted to high-girth graphs were given in~\cite{FichtenbergerPS15}.

Obtaining a characterisation of constant query testable properties is a long-standing open problem.
Ito et al.~\cite{ito2020characterization} give a characterisation of the 1-sided error constant query testable monotone and hereditary graph properties in the bounded degree (directed and undirected) graph model. Fichtenberger et al.~\cite{fichtenberger2019every} show that every constant query testable property in the bounded degree graph model is either finite or contains an infinite hyperfinite subproperty.

\subparagraph*{Organisation.}
In Section \ref{sec: prelims} we introduce relevant notions used throughout the paper. In Section \ref{sec: the model} we introduce our property testing model for bounded degree relational databases and we compare it to the classical model. In Section \ref{sec: main results} we prove our main theorems.  Due to space constraints the proofs of statements labelled $(\ast)$ are deferred to the appendix.

\section{Preliminaries}\label{sec: prelims}
We let $\mathbb{N}$ be the set of natural numbers including $0$, and $\mathbb{N}_{\geq 1} = \mathbb{N} \setminus \{0\}$. For each $n \in \mathbb{N}_{\geq 1}$, we let $[n] = \{1,2,\dots,n\}$.

\subparagraph*{Databases.}
A \emph{schema} is a finite set $\sigma = \{R_1,\dots,R_{|\sigma|}\}$ of relation names, 
where each $R\in \sigma$ has an \emph{arity} ar$(R) \in \mathbb{N}_{\geq 1}$. 
A \emph{database} $\mathcal{D}$ of schema $\sigma$ ($\sigma$-db for short) is of the form
$\mathcal{D} = (D, R_1^{\mathcal{D}}, \dots, R_{|\sigma|}^{\mathcal{D}})$, where $D$ is a finite set, the set
 of \emph{elements} of $\mathcal{D}$, and $R_i^{\mathcal{D}}$ is an ar$(R_i)$-ary relation on $D$.
 The set $D$ is also called the \emph{domain} of $\mathcal{D}$. An \emph{(undirected) graph} $\mathcal{G}$ is a tuple $\mathcal{G} =(V(\mathcal{G}),E(\mathcal{G}))$ where $V(\mathcal{G})$ is a set of \emph{vertices} and $E(\mathcal{G})$ is a set of $2$-element subsets of $V(\mathcal{G})$ (the \emph{edges} of $\mathcal G$). 
%For an edge $\{u,v\}\in E(\mathcal G)$ we simply write $uv$.
%For a graph $\mathcal G$ with $uv\in E(\mathcal G)$ we let $\mathcal G\setminus uv$ denote the graph obtained from $\mathcal{G}$ by removing the edge $uv$ from $E({\mathcal{G})}.$
An undirected graph can be seen as a $\{E\}$-db, where $E$ is a binary relation name, interpreted by a symmetric, irreflexive relation.

We assume that all databases are linearly ordered or, equivalently, that $D=[n]$ for some $n\in \mathbb N$ (similar to \cite{KazanaS11}). We extend this linear ordering to a linear order on the relations of $\mathcal{D}$ via lexicographic ordering.
The \emph{Gaifman graph} of a $\sigma$-db $\mathcal D$ is the undirected graph 
$\mathcal{G}(\mathcal{D})=(V,E)$, 
with vertex set $V:=D$ and an edge between vertices $a$ and $b$ whenever $a\neq b$ and there is an 
$R\in \sigma$ and a 
tuple $(a_1,\ldots,a_{\text{ar}(R)})\in R^{\mathcal D}$ with $a,b\in\{a_1,\ldots,a_{\text{ar}(R)}\}$.
The \emph{degree} deg$(a)$ of an element $a$ in a database $\mathcal{D}$ is the total number of tuples in all relations of $\mathcal D$ that contain $a$. We say the \emph{degree} deg$(\mathcal{D})$ of a database $\mathcal{D}$ is the maximum degree of its elements. A class of databases $\mathbf{C}$ has \emph{bounded degree}, if there exists a constant $d\in\mathbb N$ such that for all $\mathcal{D} \in \mathbf{C}$, deg$(\mathcal{D}) \leq d$. (We always assume that classes of databases are closed under isomorphism.) Let us remark that the $\deg(\mathcal{D})$ and the (graph-theoretic) degree of $\mathcal{G}(\mathcal{D})$ only differ by at most a constant factor (cf.\ e.\,g.~\cite{durand2007first}). Hence both measures yield the same classes of relational structures of bounded degree.
We define the \emph{tree-width} of a database $\mathcal D$ %$\tw(\mathcal D)$, 
as the the tree-width of its Gaifman graph. (See e.\,g.\ \cite{Flum:2006:PCT:1121738} for a discussion of tree-width in this context.)
%\cite[Proposition 11.27]{Flum:2006:PCT:1121738}. 
A class $\mathbf{C}$ of databases has \emph{bounded tree-width}, if there exists a constant $t\in \mathbb N$ such that all databases $\mathcal{D} \in \mathbf{C}$ have tree-width at most~$t$. 
%A class $\mathbf{C}$ of databases is \emph{closed under removing tuples}, if for every database $\mathcal{D} \in \mathbf{C}$, the structure $\mathcal{D'}$ obtained from $\mathcal{D} $ by deleting a tuple from some relation of $\mathcal{D} $ is also a member of $\mathbf{C}$.
Let $\mathcal D$ be a $\sigma$-db, and $M\subseteq D$. The sub-database of $\mathcal{D}$ \emph{induced by} $M$ is the database $\mathcal{D}[M]$ with domain $M$ and $R^{\mathcal{D}[M]}:=R^{\mathcal{D}}\cap M^{\text{ar}(R)}$ for every $R\in \sigma$. An \emph{$(\epsilon, k)$-partition} of a $\sigma$-db $\mathcal{D}$ on $n$ elements is a $\sigma$-db $\mathcal{D'}$ formed by removing at most $\epsilon n$ many tuples from $\mathcal{D}$ such that every connected component in $\mathcal{D'}$ contains at most $k$ elements. A class of $\sigma$-dbs $\mathbf{C} \subseteq \mathbf{D}$ is \emph{$\rho$-hyperfinite} on $\mathbf{D}$ if for every $\epsilon \in (0,1]$ and $\mathcal{D} \in \mathbf{C}$ there exists an $(\epsilon, \rho(\epsilon))$-partition $\mathcal{D'} \in \mathbf{D}$ of $\mathcal{D}$. We call $\mathbf{C}$ \emph{hyperfinite} on $\mathbf{D}$ if there exists a function $\rho$ such that $\mathbf{C}$ is $\rho$-hyperfinite on $\mathbf{D}$.

\subparagraph*{Logics.}
We shall only briefly introduce first-order logic (FO) and monadic second-order logic with counting (CMSO). Detailed introductions can be found in \cite{libkin2013elements} and \cite{courcelle2012graph}. Let \textbf{var} be a countable infinite set of \emph{variables}, and fix a relational schema $\sigma$. 
The set $\operatorname{FO}[\sigma]$ is built from \emph{atomic formulas} of the form $x_1=x_2$ or $R(x_1, \dots, x_{\textup{ar}(R)})$, where $R \in \sigma$ and $x_1,\dots,x_{\textup{ar}(R)} \in \textbf{var}$, and is closed under Boolean connectives ($\lnot, \lor,\land,\rightarrow, \leftrightarrow$) 
and existential and universal quantifications ($\exists, \forall$). 
%The set $\operatorname{FO}[\{E\}]$ is the set of first-order formulas for undirected graphs.
%We let $\operatorname{FO}:=\bigcup_{\sigma\text{ schema}}\operatorname{FO}[\sigma]$. %\emph{Two variable first-order logic} $\operatorname{FO}^2$ is the fragment of FO where formulas can be written using only two variables. We let $\operatorname{C}^2$ be the extension of $\operatorname{FO}^2$ that allows counting quantifiers $\exists^{\geq k}$ (where $\exists^{\geq k} \phi$ is true in a $\sigma$-db if the number of its elements for which $\phi$ is satisfied is at least $k$).
\emph{Monadic second-order logic} (MSO) is the extension of first-order logic that also allows quantification over subsets of the domain. CMSO extends MSO by allowing first-order modular counting quantifiers $\exists^m$ for every integer $m$ (where $\exists^m \phi$ is true in a $\sigma$-db if the number of its elements for which $\phi$ is satisfied is divisible by $m$). A \emph{free variable} of a formula is a (individual or set) variable that does not appear in the scope of a quantifier. A formula without free variables is called a \emph{sentence}. For a $\sigma$-db $\mathcal{D}$ and a sentence $\phi$ we write $\mathcal{D} \models \phi$ to denote that $\mathcal{D}$ satisfies $\phi$.

\begin{proviso}
For the rest of the paper, we fix a schema $\sigma$ and numbers $d,t \in \mathbb{N}$ with $d \geq 2$. From now on, all databases are $\sigma$-dbs and have degree at most $d$, unless stated otherwise. We use $\mathbf{C}_d$ to denote the class of all $\sigma$-dbs with degree at most $d$, $\mathbf{C}_d^t$ to denote the class of all $\sigma$-dbs with degree at most $d$ and tree-width at most $t$ and finally we use $\mathbf{C}$ to denote a class of $\sigma$-dbs with degree at most $d$.
\end{proviso}

\subparagraph*{Property testing.}
Adler and Harwath~\cite{adler2018property} introduced the model of property testing for bounded degree relational databases, which is a straightforward extension of the model for bounded degree graphs~\cite{goldreich2002property}. We call this model the \emph{\BDRD model} for short, which we shall discuss below.

Property testing algorithms do not have access to the whole input database. Instead, they are given access via an \emph{oracle}. Let $\mathcal{D}$ be an input $\sigma$-db on $n$ elements. A property testing algorithm receives the number $n$ as input, and it can make \emph{oracle queries}\footnote{Note that an oracle query is not a database query.} of the form $(R,i,j)$, where $R \in \sigma$, $i \leq n$ and $j \leq \text{deg}(\mathcal{D})$. The answer to $(R,i,j)$ is the $j^{\text{th}}$ tuple in $R^{\mathcal{D}}$ containing the $i^{\text{th}}$ element\footnote{According to the assumed linear order on $D$.} of $\mathcal{D}$ (if such a tuple does not exist then it returns $\bot$). We assume oracle queries are answered in constant time. 

Let $\mathcal{D},\mathcal{D'}$ be two $\sigma$-dbs, both having $n$ elements. In the \BDRD model the \emph{distance} between $\mathcal{D}$ and $\mathcal{D'}$, denoted by dist$(\mathcal{D}, \mathcal{D'})$, is the minimum number of tuples that have to be inserted or removed from relations of $\mathcal{D}$ and $\mathcal{D'}$ to make $\mathcal{D}$ and $\mathcal{D'}$ isomorphic. For $\epsilon \in [0,1]$, we say $\mathcal{D}$ and $\mathcal{D'}$ are \emph{$\epsilon$-close} if dist$(\mathcal{D}, \mathcal{D'}) \leq \epsilon d n$, and $\mathcal{D}$ and $\mathcal{D'}$ are \emph{$\epsilon$-far} otherwise. A \emph{property} is simply an isomorphism-closed class of databases. Note that every CMSO  sentence $\phi$ defines a property $\mathbf{P}_{\phi}=\{\mathcal D\mid \mathcal D \models \phi\}$. We call $\mathbf{P}_{\phi}\cap \mathbf{C}$ the property \emph{defined by $\phi$ on $\mathbf{C}$}.
A $\sigma$-db $\mathcal{D}$ is \emph{$\epsilon$-close} to a property $\mathbf{P}$ if there exists a database $\mathcal{D'} \in \mathbf{P}$ that is $\epsilon$-close to $\mathcal{D}$, otherwise $\mathcal{D}$ is \emph{$\epsilon$-far} from $\mathbf{P}$.

Let $\mathbf{P} \subseteq \mathbf{C}$ be a property and $\epsilon \in (0,1]$ be the proximity parameter. An \emph{$\epsilon$-tester} for $\mathbf{P}$ on $\mathbf{C}$ is a probabilistic algorithm which is given oracle access to a $\sigma$-db $\mathcal{D} \in \mathbf{C}$ and it is given $n:=|D|$ as auxiliary input. The algorithm does the following:
\begin{enumerate}
\item If $\mathcal{D} \in \mathbf{P}$, then the tester accepts with probability at least ${2}/{3}$.
\item If $\mathcal{D}$ is $\epsilon$-far from $\mathbf{P}$, then the tester rejects with probability at least ${2}/{3}$.
\end{enumerate}
The \emph{query complexity} of a tester is the maximum number of oracle queries made. 
A tester has \emph{constant} query complexity, if the query complexity does not depend on
the size of the input database.
We say a property $\mathbf{P} \subseteq \mathbf{C}$ is \emph{uniformly testable} in time $f(n)$ on $\mathbf{C}$, if for every $\epsilon \in (0,1]$ there exists an $\epsilon$-tester for $\mathbf{P}$ on $\mathbf{C}$ which has constant query complexity and whose running time on databases on $n$ elements is $f(n)$. Note that this tester must work for all $n$.

%Adler and Harwath showed that, on the class of all databases with bounded degree and tree-width, every property definable in monadic second-order logic with counting (CMSO) is uniformly testable in polylogarithmic running time~\cite{adler2018property}. (Where a function is \emph{polylogarithmic} in $n$, if it is a polynomial in $\log n$.)

\subparagraph*{Neighbourhoods.}
For a $\sigma$-db $\mathcal D$ and $a,b \in D$, the \emph{distance} between $a$ and $b$ in $\mathcal D$, denoted by dist$_{\mathcal D}(a,b)$, is the length of a shortest path between $a$ and $b$ in $\mathcal{G}(\mathcal{D})$. Let $r \in \mathbb{N}$. 
For an element $a\in D$, we let $N^{\mathcal D}_r(a)$ denote the set of all elements of $\mathcal{D}$ that are at distance at most $r$ from $a$. The \emph{$r$-neighbourhood} of $a$ in $\mathcal D$, denoted by $\mathcal{N}^{\mathcal D}_r(a)$, is the tuple $(\mathcal{D}[N_r(a)], a)$ where $a$ is called the \emph{centre}. We omit the superscript and
write $N_r(a)$ and $\mathcal{N}_r(a)$, if $\mathcal D$ is clear from the context.
Two $r$-neighbourhoods, $\mathcal{N}_r(a)$ and $\mathcal{N}_r(b)$, are \emph{isomorphic} (written $\mathcal{N}_r(a) \cong  \mathcal{N}_r(b)$) if there is an isomorphism between $\mathcal{D}[N_r(a)]$ and $\mathcal{D}[N_r(b)]$ which maps $a$ to $b$. An $\cong$-equivalence-class of $r$-neighbourhoods is called an \emph{$r$-neighbourhood type} (or \emph{$r$-type} for short). We let $T_{r}^{\sigma, d}$ denote the set of all $r$-types with degree at most $d$, over schema $\sigma$. Note that for fixed $d$ and $\sigma$, the cardinality 
$|T_{r}^{\sigma, d}|=:\operatorname{c}(r)$ is a constant, only depending on $r$ and $d$. 
%(depending only on $\sigma$ and $d$). 
We say that an element $a\in D$ \emph{has $r$-type $\tau$}, if $\mathcal{N}_r^{\mathcal D}(a) \in \tau$. 
For $r\in \mathbb N$, the \emph{$r$-histogram} of a database $\mathcal{D}$,
denoted by $\operatorname{h}_r(\mathcal{D})$, is the vector with
$\operatorname{c}(r)$ components, indexed by the $r$-types, where the
component corresponding to type $\tau$ contains the number of elements of
$\mathcal{D}$ of $r$-type $\tau$. The \emph{$r$-neighbourhood distribution} of $\mathcal{D}$, denoted by $\operatorname{dv}_r(\mathcal{D})$, is the vector $\operatorname{h}_r(\mathcal{D})/n$ where $|D|=n$. 
For a class of $\sigma$-dbs $\mathbf{C}$ and $r \in \mathbb{N}$, we let $\operatorname{h}_r(\mathbf{C}) := \{\operatorname{h}_r(\mathcal{D}) \mid \mathcal{D} \in \mathbf{C} \}$. A set is \emph{semilinear} if it is a finite union of linear sets. A set $M \subseteq \mathbb{N}^c$ is linear if $M = \{\bar{v}_0 + a_1 \bar{v}_1 + \dots + a_k \bar{v}_k \mid a_1 ,\dots, a_k \in \mathbb{N}\}$, for some $\bar{v}_0 ,\dots, \bar{v}_k \in \mathbb{N}^c$. From a result in \cite{fischer2004spectra} about many-sorted spectra of CMSO sentences it can be derived that that the set of $r$-histograms of properties defined by a CMSO sentence on $\mathbf{C}_d^t$ are semilinear.

\begin{lemma}[\cite{adler2018property,fischer2004spectra}]\label{lemma:CMSO semilinear}
For each $r \in \mathbb{N}$ and each property $\mathbf{P} \subseteq \mathbf{C}_d^t$ definable by a CMSO sentence on $\mathbf{C}_d^t$, the set $\operatorname{h}_r(\mathbf{P})$ is semilinear.
\end{lemma}

\subparagraph*{Model of computation.}
We use Random Access Machines (RAMs) and a uniform cost measure when analysing our algorithms, i.\,e.\ we assume all basic arithmetic operations including random sampling can be done in constant time, regardless of the size of the numbers involved. 

\section{The Model}\label{sec: the model}

We shall now introduce our property testing model for bounded degree relational databases, which is an extension of the \BDRD model discussed in Section \ref{sec: prelims}. The notions of oracle queries, properties, $\epsilon$-tester, query complexity and uniform testability remain the same but we have an alternative definition of distance and $\epsilon$-closeness. In our model, which we shall call the \emph{\BDRDnew model} for short, we can add and remove elements as well as tuples and can therefore compare databases that are on a different number of elements.

\begin{definition}[Distance and $\epsilon$-closeness]
	Let $\mathcal{D}, \mathcal{D'} \in \mathbf{C}_d$ and $\epsilon \in [0,1]$. The distance between $\mathcal{D}$ and $\mathcal{D'}$ (denoted by $\operatorname{dist}_{+/-}(\mathcal{D}, \mathcal{D'})$) is the minimum number of modifications we need to make to $\mathcal{D}$ and $\mathcal{D'}$ to make them isomorphic where a modification is either (1) inserting a new element, (2) deleting an element (and as a result deleting any tuple that contains that element), (3) inserting a tuple, or (4) deleting a tuple. We then say $\mathcal{D}$ and $\mathcal{D'}$ are $\epsilon$-close if $\operatorname{dist}_{+/-}(\mathcal{D}, \mathcal{D'}) \leq \epsilon d \operatorname{min}\{|D|,|D'|\}$ and are $\epsilon$-far otherwise.
\end{definition}
The following example illustrates the difference between the distance measure of the \BDRD and the distance measure of the \BDRDnew model.

\begin{example}\label{example: model comparison}

Let $\mathbf{P}=\{\mathcal{G}_{n,m} \mid n,m \in \mathbb{N}_{>1}\}$ where $\mathcal{G}_{n,m}$ is an $n$ by $m$ grid graph as shown in Figure \ref{fig: grids example}. Let us consider the graph $\mathcal{H}_{n,m}$ for some $n,m \in \mathbb{N}$ which is formed from $\mathcal{G}_{n,m}$ by removing a corner vertex. In the \BDRDnew model the distance between $\mathcal{H}_{n,m}$ and $\mathcal{G}_{n,m}$ is 1 (we remove a corner vertex from $\mathcal{G}_{n,m}$ to get $\mathcal{H}_{n,m}$) and therefore $\mathcal{H}_{n,m}$ is at distance 1 from $\mathbf{P}$ in the \BDRDnew model. In the \BDRD model if two graphs are on a different number of vertices then the distance between them is infinity. Therefore if $nm-1$ is a prime number then $\mathcal{H}_{n,m}$ is at distance infinity from $\mathbf{P}$ in the \BDRD model.

%Let $P=\{\mathcal{G}_{n,m} \mid n,m \in \mathbb{N}\}$ where $\mathcal{G}_{n,m}$ is an $n$ by $m$ grid graph. Let us consider the graph $\mathcal{H}_{n,n}$ for some $n \in \mathbb{N}$ which is formed from $\mathcal{G}_{n,n}$ by removing a corner vertex. In the \BDRDnew model the distance between $\mathcal{H}_{n,n}$ and $\mathcal{G}_{n,n}$ is 1 (we remove a corner vertex from $\mathcal{G}_{n,n}$ to get $\mathcal{H}_{n,n}$) and therefore $\mathcal{H}_{n,n}$ is at distance 1 from $P$ in the \BDRDnew model. 
%
%In the \BDRD model if two graphs are on a different number of vertices then the distance between them is infinity. The closest graph in $P$ to $\mathcal{H}_{n,n}$ is $\mathcal{G}_{n + 1,n - 1}$ (note they both have $n^2-1$ many vertices). To make $\mathcal{H}_{n,n}$ isomorphic to $\mathcal{G}_{n + 1,n - 1}$ we need to move the bottom $n-1$ vertices to one of the sides as shown in Figure \ref{fig: grids example}. This requires at least $2(n-2)$ modifications and so the distance between $\mathcal{H}_{n,n}$ and $\mathcal{G}_{n + 1,n - 1}$ (and therefore $P$) is $2(n-2)$\polly{is there any other ways?}.
%
%
\begin{figure}
\begin{center}
\begin{tikzpicture}
[scale=.6,auto=left,every node/.style={circle,fill=black!,scale=.6}]

\def\maxX{5}
\def\maxY{5}

\foreach \x  in {0,...,5}{
\foreach \y in {0,...,5}{
\ifthenelse{\x = 3 \OR \y=2}
{}
{\node at (\x,\y) {};};

\ifthenelse{\x=2 \AND \NOT \y =2}{
\draw[loosely dotted, line width=1pt] (\x,\y) -- (\x+2,\y);}
{};

\ifthenelse{\y=3 \AND \NOT \x =3}{
\draw[loosely dotted, line width=1pt] (\x,\y) -- (\x,\y -2);}
{};

\ifthenelse{ \x=5 \OR  \x=2 \OR \x=3 \OR \y=2}{}{
\draw[line width=.5pt] (\x,\y) -- (\x +1,\y);};

\ifthenelse{ \y=5 \OR  \y=1 \OR \y=2 \OR \x=3}{}{
\draw[line width=.5pt] (\x,\y) -- (\x,\y +1);};

}}

\draw [<->, line width=.7pt] ( -.75,0) -- (-.75,5);
    \node[style={fill=none, scale=1.7},rotate=90] at ( -1.1,2.5) {$n$};
    \draw [<->, line width=.7pt] (0 ,5.75) -- (5 ,5.75);
    \node[style={fill=none, scale=1.7}] at (2.5 ,6) {$m$};
   
\foreach \x  in {8,...,13}{
\foreach \y in {0,...,5}{
\ifthenelse{\x = 11 \OR \y=2 \OR \(\x=13 \AND \y=0\)}
{}
{\node at (\x,\y) {};};

\ifthenelse{\x=10 \AND \NOT \y =2}{
\draw[loosely dotted, line width=1pt] (\x,\y) -- (\x+2,\y);}
{};

\ifthenelse{\y=3 \AND \NOT \x =11}{
\draw[loosely dotted, line width=1pt] (\x,\y) -- (\x,\y -2);}
{};

\ifthenelse{ \x=13 \OR  \x=10 \OR \x=11 \OR \y=2 \OR \(\x=12 \AND \y=0\)}{}{
\draw[line width=.5pt] (\x,\y) -- (\x +1,\y);};

\ifthenelse{ \y=5 \OR  \y=1 \OR \y=2 \OR \x=11 \OR \(\x=13 \AND \y=0\)}{}{
\draw[line width=.5pt] (\x,\y) -- (\x,\y +1);};

}}

\draw [<->, line width=.7pt] ( 7.25,0) -- (7.25,5);
    \node[style={fill=none, scale=1.7},rotate=90] at ( 6.9,2.5) {$n$};
    \draw [<->, line width=.7pt] (8 ,5.75) -- (13 ,5.75);
    \node[style={fill=none, scale=1.7}] at (10.5 ,6) {$m$};

\end{tikzpicture}
\end{center}
\caption{The graphs $\mathcal{G}_{n,m}$ and $\mathcal{H}_{n,m}$ (respectively) of Example \ref{example: model comparison}.}
\label{fig: grids example}
\end{figure}
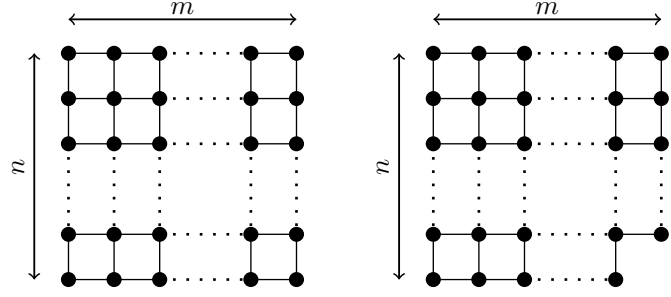
\end{example}

We now show that if a property is testable in the \BDRD model then it is also testable in the \BDRDnew model but the converse is not true. This allows for more testable properties in the \BDRDnew model.

\begin{lemma}[$\ast$]\label{lemma: comparing models 1} Let $\mathbf{P} \subseteq \mathbf{C}$. If $\mathbf{P}$ is uniformly testable on $\mathbf{C}$ in time $f(n)$ in the \BDRD model then $\mathbf{P}$ is also uniformly testable on $\mathbf{C}$ in time $f(n)$ in the \BDRDnew model.
\end{lemma}

\begin{theorem}[\cite{goldreich2002property}]\label{thrm: bipartite tester}
	In the bounded degree model, bipartiteness cannot be tested with query complexity 
	$o(\sqrt{n})$, where $n$ is the number of vertices of the input graph.
\end{theorem}

\begin{lemma}\label{lemma: comparing models 2}  There exists a class $\mathbf{C}$ of $\sigma$-dbs and a property $\mathbf{P} \subseteq \mathbf{C}$ such that $\mathbf{P}$ is trivially testable on $\mathbf{C}$ in the \BDRDnew model but is not testable on $\mathbf{C}$ in the \BDRD model.
\end{lemma}
\begin{proof}
	Let $\mathbf{C}$ be the class of all graphs with degree at most $d$. Let $\mathbf{P} = \mathbf{P_1} \cup \mathbf{P_2} \subseteq \mathbf{C}$ be the property where $\mathbf{P_1}$ contains all bipartite graphs in $\mathbf{C}$ and $\mathbf{P_2}$ contains all graphs in $\mathbf{C}$ that have an odd number of vertices. In the \BDRDnew model every $\mathcal{G} \in \mathbf{C}$ is $\epsilon$-close to $\mathbf{P}$ if $|V(\mathcal{G})|\geq 1/(\epsilon d) $
 and hence $\mathbf{P}$ is trivially testable on $\mathbf{C}$ in the \BDRDnew model (the tester accepts if $|V(\mathcal{G})|\geq 1/(\epsilon d) $ and does a full check of the input otherwise). In the \BDRD model, if the input graph has an even number of vertices then it is far from $\mathbf{P_2}$ and so we have to test for $\mathbf{P_1}$. By Theorem \ref{thrm: bipartite tester}, bipartiteness is not testable (with constant query complexity) in the \BDRD model. In particular, in the proof of Theorem \ref{thrm: bipartite tester}, Goldreich and Ron show that for any even $n$ there exists two families, $\mathcal{G}_1 \subseteq \mathbf{C}$ and $\mathcal{G}_2 \subseteq \mathbf{C}$, of $n$-vertex graphs such that every graph in $\mathcal{G}_1 $ is bipartite and almost all graphs in $\mathcal{G}_2 $ are far from being bipartite but any algorithm that performs $o(\sqrt{n})$ queries cannot distinguish between a graph chosen randomly from $\mathcal{G}_1 $ and a graph chosen randomly from $\mathcal{G}_2 $. Therefore $\mathbf{P}$ is not testable on $\mathbf{C}$ in the \BDRD model.
%Note for any input class $\mathbf{C}$ and properties $\mathbf{P_1}, \mathbf{P_2} \subseteq \mathbf{C}$, if $\mathbf{P_1}$ is not testable on $\mathbf{C} \setminus \mathbf{P_2}$ in the \BDRD model, and every input in $\mathbf{C} \setminus \mathbf{P_2}$ is $\epsilon$-far from $\mathbf{P_2}$ in the \BDRD model but $\epsilon$-close in the \BDRDnew model, then $\mathbf{P_1} \cup \mathbf{P_2}$ is not testable on $\mathbf{C}$ in the \BDRD model but is testable in the \BDRDnew model. \polly{added this, check it is correct}
%Let $\mathbf{C} = \mathbf{C}_d$ and let $P = P_1 \cup P_2 \subseteq \mathbf{C}$ where $P_1$ is a property that is not testable on $\mathbf{C}$ in the \BDRD model and $P_2$ contains all $\sigma$-dbs in $\mathbf{C}$ on an even (or alternatively odd) number of elements. In the \BDRDnew model every $\mathcal{D} \in \mathbf{C}$ is $\epsilon$-close to $P$ if $1 \leq \epsilon d |D|$ and hence $P$ is trivially testable on $\mathbf{C}$ in the \BDRDnew model. In the \BDRD model, if the input $\sigma$-db is not on an even (or alternatively odd) number of elements then it is far from $P_2$ and so we have to test for $P_1$. Since $P_1$ is not testable on $\mathbf{C}$, $P$ is also not testable on $\mathbf{C}$ in the \BDRD model.
\end{proof}

Note that the underlying general principle of the above proof can be applied to obtain further examples of properties that are testable in the \BDRDnew model but not testable in the \BDRD model.

It is known that every hyperfinite property is `local' (Theorem~\ref{thrm: local-global}), where `local' means that if a $\sigma$-db $\mathcal{D}$ has a similar $r$-histogram to some $\sigma$-db (with the same domain size) that has the (hyperfinite) property, then $\mathcal{D}$ must be $\epsilon$-close to the property~\cite{newman2013every,adler2018property}. This is summarised in Theorem~\ref{thrm: local-global} below.
%The following theorem is a generalisation of \cite[Theorem 3.1]{newman2013every} from graphs to structures. 
We use Theorem~\ref{thrm: local-global} to prove a similar result in the \BDRDnew model (Lemma~\ref{lemma: locality}).  Lemma~\ref{lemma: locality} is essential for the proof of Theorem \ref{thrm: small dbs}.

\begin{theorem}[\cite{newman2013every,adler2018property}]\label{thrm: local-global}
Let $\epsilon \in (0,1]$ and let $\mathbf{C}$ be closed under removing tuples. If a property $\mathbf{P} \subseteq \mathbf{C}$ is hyperfinite on $\mathbf{C}$ then there exists $\lambda_{\ref{thrm: local-global}} := \lambda_{\ref{thrm: local-global}} (\epsilon) \in (0,1]$ and $r_{\ref{thrm: local-global}} := r_{\ref{thrm: local-global}}(\epsilon) \in \mathbb{N}$ such that for each $\mathcal{D} \in \mathbf{P}$ and $ \mathcal{D'} \in \mathbf{C}$ with the same number $n$ of elements, if $\|h_{r_{\ref{thrm: local-global}}}(\mathcal{D})- h_{r_{\ref{thrm: local-global}}}(\mathcal{D'})\|_1 \leq \lambda_{\ref{thrm: local-global}} n$, then $\mathcal{D'}$ is $\epsilon$-close to $\mathbf{P}$ in the \BDRD model.
\end{theorem}
\begin{lemma}\label{lemma: locality}
Let $\epsilon \in (0,1]$ and let $\mathbf{C}$ be closed under removing tuples. If a property $\mathbf{P} \subseteq \mathbf{C}$ is hyperfinite on $\mathbf{C}$ then there exists $\lambda := \lambda (\epsilon) \in (0,1]$ and $r := r(\epsilon) \in \mathbb{N}$ such that for each $\mathcal{D} \in \mathbf{P}$ and $ \mathcal{D'} \in \mathbf{C}$, on $|D|$ and $|D'|$ elements respectively, if $\|\operatorname{h}_r(\mathcal{D})- \operatorname{h}_r(\mathcal{D'})\|_1 \leq \lambda \operatorname{min}\{|D|,|D'|\}$, then $\mathcal{D'}$ is $\epsilon$-close to $\mathbf{P}$ in the \BDRDnew model.
\end{lemma}
\begin{proof}
Let $r =r_{\ref{thrm: local-global}}(\epsilon/4)$ and let $\lambda = \frac{\epsilon\lambda_{\ref{thrm: local-global}} (\epsilon /4)}{1+d^{r+1}}$. Let us assume that $\|\operatorname{h}_r(\mathcal{D})- \operatorname{h}_r(\mathcal{D'})\|_1 \leq \lambda  \operatorname{min}\{|D|,|D'|\}$ and $\mathbf{P}$ is hyperfinite on $\mathbf{C}$. If $|D|=|D'|$ then by Theorem~\ref{thrm: local-global} and the choice of $\lambda$, $\mathcal{D'}$ is $\epsilon$-close to $\mathbf{P}$. So let us assume that $|D|\neq|D'|$. Let $\mathcal{D}_1$ be the $\sigma$-db on $|D|$ elements formed from $\mathcal{D'}$ by either removing $|D'|-|D|$ elements if $|D| < |D'|$ or adding $|D|-|D'|$ new elements if $|D'| < |D|$. 
Note that as $\|\operatorname{h}_r(\mathcal{D})- \operatorname{h}_r(\mathcal{D'})\|_1 \leq \lambda  \operatorname{min}\{|D|,|D'|\}$ and by definition $\|\operatorname{h}_r(\mathcal{D})- \operatorname{h}_r(\mathcal{D'})\|_1 = \sum_{i=1}^{\operatorname{c}(r)}|\operatorname{h}_r(\mathcal{D})- \operatorname{h}_r(\mathcal{D'})|$ we have $\big| |D|-|D'| \big| \leq \lambda  \operatorname{min}\{|D|,|D'|\}$.
When an element $a$ is removed, the $r$-type of any element in $N_r(a)$ will change. 
	As $|N_r(a)| \leq d^{r+1}$ (cf.\ e.\,g.\ Lemma 3.2 (a) of \cite{berkholz2018answering}) and $\big| |D|-|D'| \big| \leq \lambda  \operatorname{min}\{|D|,|D'|\}$, we have $\|\operatorname{h}_r(\mathcal{D'})- \operatorname{h}_r(\mathcal{D}_1)\|_1 \leq \lambda  \operatorname{min}\{|D|,|D'|\}d^{r+1}$. 
%Recall that for a vector $\bar{v}$ on $\ell$ components, $\| \bar{v} \|_1 := \sum_{i=1}^{\ell} |\bar{v}[i]|$. 
%	Then, as $\|\operatorname{h}_r(\mathcal{D})- \operatorname{h}_r(\mathcal{D'})\|_1 \leq \lambda  \operatorname{min}\{|D|,|D'|\}$, we have $\big| |D|-|D'| \big| \leq \lambda  \operatorname{min}\{|D|,|D'|\}$. 
Therefore \[\|\operatorname{h}_r(\mathcal{D})- \operatorname{h}_r(\mathcal{D}_1)\|_1 \leq \lambda  \operatorname{min}\{|D|,|D'|\}(1 +d^{r+1}) \leq \lambda_{\ref{thrm: local-global}} (\epsilon /4) |D|\] by the choice of $\lambda$. By Theorem \ref{thrm: local-global}, in the \BDRD model $\mathcal{D}_1$ is $\epsilon /4$-close to $\mathbf{P}$. Hence there exists a $\sigma$-db $\mathcal{D}_{2} \in \mathbf{P}$ such that $|D_2|=|D|$ and $\operatorname{dist}(\mathcal{D}_1, \mathcal{D}_2) \leq \epsilon d |D|/4$. By the definition of the two distance measures $\operatorname{dist}$ and $\operatorname{dist}_{+/-}$, we have $\operatorname{dist}_{+/-}(\mathcal{D}_1, \mathcal{D}_2) \leq \operatorname{dist}(\mathcal{D}_1, \mathcal{D}_2)\leq \epsilon d |D|/4$ and by the choice of $\mathcal{D}_1$ we have $\operatorname{dist}_{+/-}(\mathcal{D'}, \mathcal{D}_1) \leq \lambda  \operatorname{min}\{|D|,|D'|\}$. Therefore \[\operatorname{dist}_{+/-}(\mathcal{D'}, \mathcal{D}_2) \leq \frac{\epsilon d |D|}{4} + \lambda  \operatorname{min}\{|D|,|D'|\} \leq \epsilon d \operatorname{min}\{|D|,|D'|\},\] as $|D| \leq \operatorname{min}\{|D|,|D'|\} + \lambda  \operatorname{min}\{|D|,|D'|\} \leq 2 \operatorname{min}\{|D|,|D'|\} $ and $\lambda \leq \epsilon d /2$. Hence in the \BDRDnew model $\mathcal{D}'$ is $\epsilon$-close to $\mathbf{P}$.
\end{proof}

\section{Main Results}\label{sec: main results}
We begin this section with the first of our main theorems (Theorem~\ref{thrm: small dbs}). We show that for any property $\mathbf{P}$ which is hyperfinite on the input class $\mathbf{C}$, if the set of $r$-histograms of $\mathbf{P}$ is semilinear, then for every $\sigma$-db $\mathcal{D}$ in $\mathbf{P}$ there exists a constant size $\sigma$-db in $\mathbf{P}$ with a neighbourhood distribution similar to that of $\mathcal{D}$, but this is not true for $\sigma$-dbs in $\mathbf{C}$ that are far from $\mathbf{P}$.
We then use this result to prove that such properties are testable in constant time in the \BDRDnew model (Theorem~\ref{thm: constant time tester}).
As a corollary we obtain that CMSO definable properties on $\sigma$-dbs of bounded tree-width and bounded degree are testable in constant time (Theorem~\ref{thm: CMSO testability}).

\begin{theorem}\label{thrm: small dbs}
 Let $\epsilon \in (0,1]$ and let $r := r(\epsilon)$ be as in Lemma \ref{lemma: locality}.  Let $\mathbf{C}$ be closed under removing tuples and let $\mathbf{P} \subseteq \mathbf{C}$ be a property that is hyperfinite on $\mathbf{C}$ such that the set $\operatorname{h}_r(\mathbf{P})$ is semilinear. There exist $n_{\text{min}}:=n_{\text{min}}(\epsilon),n_{\text{max}}:=n_{\text{max}}(\epsilon) \in \mathbb{N}$ and $f:=f(\epsilon), \mu:=\mu(\epsilon) \in (0,1)$ such that for every $\mathcal{D} \in \mathbf{C}$ with $|D| > n_{\text{max}}$,
\begin{enumerate}
\item if $\mathcal{D} \in \mathbf{P}$, then there exists a $\mathcal{D'} \in \mathbf{P}$ such that $n_{\text{min}} \leq |D'| \leq n_{\text{max}}$ and $\|\operatorname{dv}_r(\mathcal{D})-\operatorname{dv}_r(\mathcal{D'})\|_1 \leq f - \mu $, and
\item if $\mathcal{D}$ is $\epsilon$-far from $\mathbf{P}$ (in the \BDRDnew model), then for every $\mathcal{D'} \in \mathbf{P}$ such that $n_{\text{min}} \leq |D'| \leq n_{\text{max}}$, we have $\|\operatorname{dv}_r(\mathcal{D}) -\operatorname{dv}_r(\mathcal{D'})\|_1 > f + \mu$.
\end{enumerate}
\end{theorem}
\begin{proof}
	Let $\lambda := \lambda(\epsilon)$ be as in Lemma \ref{lemma: locality} and $c:=\operatorname{c}(r)$ (the number of $r$-types). First note that if $\mathbf{P}$ is empty then for any choice of $n_{\text{min}}$, $n_{\text{max}}$, $f$ and $\mu$, both 1. and 2. in the theorem statement are true and hence we shall assume that $\mathbf{P}$ is non-empty. As $\operatorname{h}_r(\mathbf{P})$ is a semilinear set we can write it as follows, $\operatorname{h}_r(\mathbf{P}) = M_1 \cup M_2 \cup \dots \cup M_m$ where $m \in \mathbb{N}$ and for each $i \in [m]$, $M_i = \{ \bar{v}_0^i + a_1 \bar{v}_1^i + \dots + a_{k_i} \bar{v}_{k_i}^i \mid a_1,\dots,a_{k_i} \in \mathbb{N} \}$ is a linear set where $\bar{v}_0^i,\dots,\bar{v}_{k_i}^i \in \mathbb{N}^{c}$ and for each $j \in [k_i]$, $\|\bar{v}_j^i\|_1 \neq 0$. Let $k:=\max_{i \in [m]}k_i + 1$ and $v:=\max_{i \in [m]} \Big(\max_{j \in [0,k_i]}\|\bar{v}_j^i\|_1\Big)$ (note that $v >0$ as $\mathbf{P}$ is non-empty). Let $n_{\text{min}}:= n_0 - kv$, $n_{\text{max}}:= n_0 + kv$, $f := \frac{\lambda}{3c}$, and $\mu := \frac{\lambda}{6c}$
%\begin{itemize}
%\item $n_{\text{min}}= n_0(1-a) - kv$,
%\item $n_{\text{max}}= n_0 + kv$,
%\item $f = \frac{\lambda}{3\operatorname{c}}$, and
%\item $\mu = \frac{\lambda}{6\operatorname{c}}$
%\end{itemize}
 where \[n_0:=kv\Big(\frac{3ckv}{f-\mu} + 1 \Big).\]
 Note that $n_{\text{min}} > 0$ by the choice of $n_0$, $f$ and $\mu$.

	(Proof of 1.) Assume $\mathcal{D} \in \mathbf{P}$ and $|D|=n > n_{\text{max}}$. Then there exists some $i \in [m]$ and $a_1^{\mathcal{D}},\dots,a_{k_i}^{\mathcal{D}} \in \mathbb{N}$ such that $\operatorname{h}_r(\mathcal{D}) = \bar{v}_0^i + a_1^{\mathcal{D}} \bar{v}_1^i + \dots + a_{k_i}^{\mathcal{D}}\bar{v}_{k_i}^i$ (note that $n = \|\bar{v}_0^i\|_1+ \sum_{j \in [k_i]} a_j^{\mathcal{D}} \|\bar{v}_j^i\|_1 $). Let $\mathcal{D'}$ be the $\sigma$-db with $r$-histogram $\bar{v}_0^i + a_1^{\mathcal{D'}} \bar{v}_1^i + \dots + a_{k_i}^{\mathcal{D'}}\bar{v}_{k_i}^i \in M_i$ where $a_j^{\mathcal{D'}}$ is the nearest integer to $a_j^{\mathcal{D}} n_{0} /n$, and hence $ a_j^{\mathcal{D}} n_0/n -1/2 \leq a_j^{\mathcal{D'}} \leq a_j^{\mathcal{D}} n_0/n + 1/2$. Note that since $\bar{v}_0^i + a_1^{\mathcal{D'}} \bar{v}_1^i + \dots + a_{k_i}^{\mathcal{D'}}\bar{v}_{k_i}^i \in \operatorname{h}_r(\mathbf{P})$, $\mathcal{D'}$ exists and $\mathcal{D'} \in \mathbf{P}$. We need to show that $n_{\text{min}} \leq |D'| \leq n_{\text{max}}$ and $\|\operatorname{dv}_r(\mathcal{D}) -\operatorname{dv}_r(\mathcal{D'})\|_1 \leq f - \mu$.

\begin{claim}[$\ast$]\label{claim: lower bound on size}
$|D'| \geq n_{\text{min}}$.
\end{claim}

%\begin{claimproof}
%\begin{align*}
%|D'| = |\bar{v}_0^i|+ a_1' |\bar{v}_1^i| + \dots a_{k_i}'|\bar{v}_{k_i}^i|  
%\geq |\bar{v}_0^i|+ \Big(\frac{a_1n_0}{n} - \frac{1}{2}\Big) |\bar{v}_1^i| + \dots \Big(\frac{a_{k_i}n_0}{n} - \frac{1}{2}\Big) |\bar{v}_{k_i}^i| 
%\end{align*}
%as for every $j \in [k_i]$, by the choice of $a_j'$, we have $  a_j n_0/n -1/2 \leq a_j' $. Then
%\begin{align*}
%&|\bar{v}_0^i|+ \Big(\frac{a_1n_0}{n} - \frac{1}{2}\Big) |\bar{v}_1^i| + \dots + \Big(\frac{a_{k_i}n_0}{n} - \frac{1}{2}\Big) |\bar{v}_{k_i}^i| \\
%&= |\bar{v}_0^i| - \frac{1}{2} \sum_{1 \leq j \leq k_i}|\bar{v}_j^i| + \frac{n_0 \sum_{1 \leq j \leq k_i}a_j|\bar{v}_j^i|}{n} \\
%&= |\bar{v}_0^i| - \frac{1}{2} \sum_{1 \leq j \leq k_i}|\bar{v}_j^i| + n_0(1-\frac{|\bar{v}_0^i|}{n}) 
%\end{align*}
%as $\sum_{1 \leq j \leq k_i}a_j|\bar{v}_j^i| = n - |\bar{v}_0^i|$.
%\begin{align*}
%&|\bar{v}_0^i| - \frac{1}{2} \sum_{1 \leq j \leq k_i}|\bar{v}_j^i| + n_0(1-\frac{|\bar{v}_0^i|}{n}) \\
%&\geq  - \sum_{0 \leq j \leq k_i}|\bar{v}_j^i| + n_0(1-a) \\
%&\geq -kv +  n_0(1-a) = n_{\text{min}}
%\end{align*}
%as $n \geq |\bar{v}_0^i|/a$.
%\end{claimproof}

\begin{claim}[$\ast$]\label{claim:upper bound on size}
$|D'| \leq n_{\text{max}}$.
\end{claim}

%\begin{claimproof}
%\begin{align*}
%|D'| = |\bar{v}_0^i|+ a_1' |\bar{v}_1^i| + \dots a_{k_i}'|\bar{v}_{k_i}^i| 
%\leq |\bar{v}_0^i|+ \Big(\frac{a_1n_0}{n} + \frac{1}{2}\Big) |\bar{v}_1^i| + \dots \Big(\frac{a_{k_i}n_0}{n} + \frac{1}{2}\Big) |\bar{v}_{k_i}^i| 
%\end{align*}
%as or every $j \in [k_i]$, by the choice of $a_j'$, we have $  a_j n_0/n +1/2 \geq a_j' $. Then
%\begin{align*}
%&|\bar{v}_0^i|+ \Big(\frac{a_1n_0}{n} + \frac{1}{2}\Big) |\bar{v}_1^i| + \dots \Big(\frac{a_{k_i}n_0}{n} + \frac{1}{2}\Big) |\bar{v}_{k_i}^i| \\
%&= |\bar{v}_0^i| + \frac{1}{2} \sum_{1 \leq j \leq k_i}|\bar{v}_j^i| + \frac{n_0 \sum_{1 \leq j \leq k_i}a_j|\bar{v}_j^i|}{n} \\
%&= |\bar{v}_0^i| + \frac{1}{2} \sum_{1 \leq j \leq k_i}|\bar{v}_j^i| + n_0(1-\frac{|\bar{v}_0^i|}{n}) 
%\end{align*}
%as $\sum_{1 \leq j \leq k_i}a_j|\bar{v}_j^i| = n - |\bar{v}_0^i|$.
%\begin{align*}
%|\bar{v}_0^i| + \frac{1}{2} \sum_{1 \leq j \leq k_i}|\bar{v}_j^i| + n_0(1-\frac{|\bar{v}_0^i|}{n}) 
%\leq   \sum_{0 \leq j \leq k_i}|\bar{v}_j^i| + n_0 
%\leq kv +  n_0 = n_{\text{max}}
%\end{align*}
%as $1-|\bar{v}_0^i|/n < 1$.
%\end{claimproof}

\begin{claim}
$\|\operatorname{dv}_r(\mathcal{D}) -\operatorname{dv}_r(\mathcal{D'})\|_1 \leq f - \mu$.
\end{claim}
\begin{claimproof}
	By definition, $\|\operatorname{dv}_r(\mathcal{D}) -\operatorname{dv}_r(\mathcal{D'})\|_1 = \sum _{j \in [c]}|\operatorname{dv}_r(\mathcal{D})[j]-\operatorname{dv}_r(\mathcal{D'})[j]|$. First recall that $0 < n_0 -kv \leq |D'| \leq n_0 +kv < n$ and note that for every $\ell \in [k_i]$, $a_{\ell}^{\mathcal{D}} \leq n$ (since $\|\bar{v}_{\ell}^i\|_1 \neq 0$). Then for every $j \in [c]$, by the choice of $a_{\ell}^{\mathcal{D'}}$  for $\ell \in [k_i]$,
\begin{align*}
 &\operatorname{dv}_r(\mathcal{D})[j]-\operatorname{dv}_r(\mathcal{D'})[j] 
 = \bar{v}_0^i[j]\Big(\frac{1}{n} - \frac{1}{|D'|}\Big) + \sum_{\ell \in [k_i]}\bar{v}_{\ell}^i[j]\Big(\frac{a_{\ell}^{\mathcal{D}}}{n} - \frac{a_{\ell}^{\mathcal{D'}}}{|D'|}\Big) \\
 &\leq \sum_{\ell \in [k_i]}\bar{v}_{\ell}^i[j]\Big(\frac{a_{\ell}^{\mathcal{D}}}{n} - \frac{a_{\ell}^{\mathcal{D}} n_0}{n|D'|} + \frac{1}{2|D'|}\Big) 
 = \sum_{\ell \in [k_i]}\bar{v}_{\ell}^i[j]\Big(\frac{a_{\ell}^{\mathcal{D}}}{n}\Big(\frac{|D'| - n_0}{|D'|}\Big) + \frac{1}{2|D'|}\Big) \\
 &\leq \sum_{\ell \in [k_i]}\bar{v}_{\ell}^i[j]\Big(\frac{n}{n}\Big(\frac{kv + n_0 - n_0}{n_0-kv}\Big) + \frac{1}{2(n_0-kv)}\Big) = \Big(\frac{2kv + 1}{2(n_0-kv)}\Big)  \sum_{\ell \in [k_i]}\bar{v}_{\ell}^i[j] \\
 &\leq \frac{kv(2kv + 1)}{n_0-kv}.
 \end{align*}
On the other hand, 
 \begin{align*}
 &\operatorname{dv}_r(\mathcal{D})[j]-\operatorname{dv}_r(\mathcal{D'})[j] 
 \geq -\frac{\bar{v}_0^i[j]}{|D'|} + \sum_{\ell \in [k_i]}\bar{v}_{\ell}^i[j]\Big(\frac{a_{\ell}^{\mathcal{D}}}{n}\Big(\frac{|D'| - n_0}{|D'|}\Big) - \frac{1}{2|D'|}\Big) \\
 &\geq -\frac{\bar{v}_0^i[j]}{|D'|} + \sum_{\ell \in [k_i]}\bar{v}_{\ell}^i[j]\Big(\frac{a_{\ell}^{\mathcal{D}}}{n}\Big(\frac{-kv + n_0 - n_0}{|D'|}\Big) - \frac{1}{2|D'|}\Big) \\
 &= -\frac{\bar{v}_0^i[j]}{|D'|} - \sum_{\ell \in [k_i]} \bar{v}_{\ell}^i[j]\Big(\frac{a_{\ell}^{\mathcal{D}}kv}{n|D'|} + \frac{1}{2|D'|}\Big) \\
 & \geq  -\frac{\bar{v}_0^i[j]}{n_0 -kv} - \sum_{\ell \in [k_i]}\bar{v}_{\ell}^i[j]\Big(\frac{nkv}{n(n_0 -kv)} + \frac{1}{2(n_0 -kv)}\Big) \\
 &= -\frac{\bar{v}_0^i[j]}{n_0 -kv} - \Big(\frac{2kv +1}{2(n_0 -kv)}\Big) \sum_{\ell \in  [k_i]}\bar{v}_{\ell}^i[j] 
 \geq -\frac{kv(2kv + 1)}{n_0-kv}.
 \end{align*}
  %as $0 < n_0(1-a) -kv \leq |D'|$ and for every $\ell \in [k_i]$,  $a_{\ell}\leq n$ (if $\bar{v}_{\ell}^i[j] \neq 0$).
  
 Hence, \[|\operatorname{dv}_r(\mathcal{D})[j]-\operatorname{dv}_r(\mathcal{D'})[j]| \leq \frac{kv(2kv + 1)}{n_0-kv} \leq \frac{3(kv)^2}{n_0-kv} =\frac{f - \mu}{c}\]
 by the choice of $n_0$.
 Therefore, \[\|\operatorname{dv}_r(\mathcal{D}) -\operatorname{dv}_r(\mathcal{D'})\|_1 = \sum _{j \in [c]}|\operatorname{dv}_r(\mathcal{D})[j]-\operatorname{dv}_r(\mathcal{D'})[j]| \leq f - \mu\] as required.
  \end{claimproof}

(Proof of 2.) Assume $\mathcal{D}$ is $\epsilon$-far from $\mathbf{P}$ and $|D|=n > n_{\text{max}}$. For a contradiction let us assume there does exist a $\sigma$-db $\mathcal{D'} \in \mathbf{P}$ such that $n_{\text{min}} \leq |D'| \leq n_{\text{max}}$ and $\|\operatorname{dv}_r(\mathcal{D}) -\operatorname{dv}_r(\mathcal{D'})\|_1 \leq f + \mu$. As $\mathcal{D'} \in \mathbf{P}$ there exists some $i \in [m]$ and $a_1^{\mathcal{D'}},\dots,a_{k_i}^{\mathcal{D'}} \in \mathbb{N}$ such that $\operatorname{h}_r(\mathcal{D'}) = \bar{v}_0^i + a_1^{\mathcal{D'}} \bar{v}_1^i + \dots + a_{k_i}^{\mathcal{D'}}\bar{v}_{k_i}^i$. Let $\mathcal{D''}$ be the $\sigma$-db with $r$-histogram $\bar{v}_0^i + a_1^{\mathcal{D''}} \bar{v}_1^i + \dots + a_{k_i}^{\mathcal{D''}}\bar{v}_{k_i}^i \in M_i$ where $a_j^{\mathcal{D''}}$ is the nearest integer to $a_j^{\mathcal{D'}} n /|D'|$. Note as $\bar{v}_0^i + a_1^{\mathcal{D''}} \bar{v}_1^i + \dots + a_{k_i}^{\mathcal{D''}}\bar{v}_{k_i}^i \in \operatorname{h}_r(\mathbf{P})$, $\mathcal{D''}$ exists and $\mathcal{D''} \in \mathbf{P}$. 
\begin{claim}\label{claim: contradiction} $\mathcal{D}$ is $\epsilon$-close to $\mathbf{P}$.
\end{claim}
\begin{claimproof}
First note that as $\|\operatorname{dv}_r(\mathcal{D}) -\operatorname{dv}_r(\mathcal{D'})\|_1  \leq f+ \mu$ and $\operatorname{h}_r(\mathcal{D'}) = \bar{v}_0^i + a_1^{\mathcal{D'}} \bar{v}_1^i + \dots + a_{k_i}^{\mathcal{D'}}\bar{v}_{k_i}^i$, for every $j \in [c]$ 
\begin{align*}
&\frac{\bar{v}_0^i[j] + \sum_{\ell \in [k_i]}a_{\ell}^{\mathcal{D'}} \bar{v}_{\ell}^i[j]}{|D'|} - f- \mu\leq \operatorname{dv}_r(\mathcal{D})[j] \leq \frac{\bar{v}_0^i[j] + \sum_{\ell \in [k_i]}a_{\ell}^{\mathcal{D'}} \bar{v}_{\ell}^i[j]}{|D'|} + f + \mu
\end{align*}
 and therefore 
\begin{align*}
&n\Big(\frac{\bar{v}_0^i[j] +  \sum_{\ell \in [k_i]}a_{\ell}^{\mathcal{D'}} \bar{v}_{\ell}^i[j]}{|D'|} - f - \mu\Big)   \leq \operatorname{h}_r(\mathcal{D})[j]\leq n \Big( \frac{\bar{v}_0^i[j] +  \sum_{\ell \in [k_i]}a_{\ell}^{\mathcal{D'}} \bar{v}_{\ell}^i[j]}{|D'|} + f + \mu \Big).
\end{align*}
Hence, by the choice of $a_{\ell}^{\mathcal{D''}}$ for $\ell \in [k_i]$,
\begin{align*}
&\operatorname{h}_r(\mathcal{D})[j] - \operatorname{h}_r(\mathcal{D''})[j] 
\leq \bar{v}_0^i[j]\Big(\frac{n}{|D'|} - 1\Big) + \sum_{\ell \in [k_i]}\bar{v}_{\ell}^i[j]\Big(\frac{a_{\ell}^{\mathcal{D'}} n}{|D'|} - a_{\ell}^{\mathcal{D''}}\Big) + fn + \mu n\\
&\leq \bar{v}_0^i[j]\frac{n}{|D'|} + \sum_{\ell \in [k_i]}\bar{v}_{\ell}^i[j]\Big(\frac{a_{\ell}^{\mathcal{D'}} n}{|D'|} - \Big(\frac{a_{\ell}^{\mathcal{D'}}n}{|D'|} - \frac{1}{2}\Big)\Big) + fn + \mu n\\
&=\bar{v}_0^i[j]\frac{n}{|D'|} + \frac{1}{2}\sum_{\ell \in [k_i]} \bar{v}_{\ell}^i[j] +fn+ \mu n.
\end{align*}
Similarly, by the choice of $a_{\ell}^{\mathcal{D''}}$ for $\ell \in [k_i]$ and as $n > |D'|$,
\begin{align*}
&\operatorname{h}_r(\mathcal{D})[j] - \operatorname{h}_r(\mathcal{D''})[j] 
\geq \bar{v}_0^i[j]\Big(\frac{n}{|D'|} - 1\Big) + \sum_{\ell \in [k_i]}\bar{v}_{\ell}^i[j]\Big(\frac{a_{\ell}^{\mathcal{D'}} n}{|D'|} - a_{\ell}^{\mathcal{D''}}\Big) - fn - \mu n\\
&\geq -\bar{v}_0^i[j]\frac{n}{|D'|} + \sum_{\ell \in [k_i]}\bar{v}_{\ell}^i[j]\Big(\frac{a_{\ell}^{\mathcal{D'}} n}{|D'|} - \Big(\frac{a_{\ell}^{\mathcal{D'}}n}{|D'|} + \frac{1}{2}\Big)\Big)  - fn - \mu n\\
&=-\bar{v}_0^i[j]\frac{n}{|D'|} - \frac{1}{2}\sum_{\ell \in [k_i]} \bar{v}_{\ell}^i[j] -f n- \mu n.
\end{align*}
Therefore,
\begin{align*}
&|\operatorname{h}_r(\mathcal{D})[j] - \operatorname{h}_r(\mathcal{D''})[j]| \leq \bar{v}_0^i[j]\frac{n}{|D'|} + \frac{1}{2}\sum_{\ell \in [k_i]} \bar{v}_{\ell}^i[j] +fn + \mu n\\
&\leq \frac{n}{|D'|}\sum_{0 \leq \ell \leq k_i} \bar{v}_{\ell}^i[j]+f n  + \mu n
\leq \frac{nkv}{|D'|} + fn + \mu n \\
 &= n\Big(\frac{kv}{|D'|} + \frac{\lambda}{3c} + \frac{\lambda}{6c} \Big)
 \leq n\Big(\frac{\lambda}{18c} + \frac{\lambda}{3c} + \frac{\lambda}{6c}\Big)
  =\frac{5 \lambda n}{9c}
\end{align*}
by the choice of $f$ and $\mu$ and as \[|D'| \geq n_{\text{min}} = \frac{3c(kv)^2}{f- \mu} =\frac{18(ckv)^2}{\lambda} \geq \frac{18ckv}{ \lambda}.\] To apply Lemma \ref{lemma: locality} we need to show that $\|\operatorname{h}_r(\mathcal{D}) - \operatorname{h}_r(\mathcal{D''})\|_1 \leq \lambda \min\{n,|D''|\}$. If $|\operatorname{h}_r(\mathcal{D})[j] - \operatorname{h}_r(\mathcal{D''})[j]| \leq \frac{\lambda}{c} \min\{n,|D''|\}$ then $\|\operatorname{h}_r(\mathcal{D}) - \operatorname{h}_r(\mathcal{D''})\|_1 \leq \lambda \min\{n,|D''|\}$. Clearly, $\frac{5 \lambda n}{9c} < \frac{\lambda n}{c}$. We also have
\begin{align*}
&|D''| = \|\bar{v}_0^i\|_1 + \sum_{\ell \in [k_i]}a_{\ell}^{\mathcal{D''}} \|\bar{v}_{\ell}^i\|_1  
\geq \|\bar{v}_0^i\|_1 + \sum_{\ell \in [k_i]}\Big(\frac{a_{\ell}^{\mathcal{D'}}n}{|D'|}-\frac{1}{2}\Big) \|\bar{v}_{\ell}^i\|_1  \\
&=\|\bar{v}_0^i\|_1 -\frac{1}{2} \sum_{\ell \in [k_i]}\|\bar{v}_{\ell}^i\|_1 + \frac{n}{|D'|}\sum_{\ell \in [k_i]}a_{\ell}^{\mathcal{D'}}\|\bar{v}_{\ell}^i\|_1
\geq -kv +\frac{n}{|D'|}(|D'| -\|\bar{v}_0^i\|_1) \\
&\geq - \frac{n}{18}+\frac{17}{18}n 
\geq \frac{5n}{9}
\end{align*}
as \[|D'| \geq \frac{18ckv}{ \lambda}  \geq 18v \geq 18\|\bar{v}_0^i\|_1 \text{ and } kv \leq \frac{(ckv)^2}{\lambda} = \frac{n_{\text{min}}}{18} \leq \frac{n}{18}.\] Therefore, $\frac{5 \lambda n}{9c} \leq \frac{\lambda |D''|}{c}$ and hence $\|\operatorname{h}_r(\mathcal{D}) - \operatorname{h}_r(\mathcal{D''})\|_1 \leq \lambda \min\{n,|D''|\}$. By Lemma \ref{lemma: locality}, $\mathcal{D}$ is $\epsilon$-close to $\mathbf{P}$.
 \end{claimproof}

Claim~\ref{claim: contradiction} gives us a contradiction and therefore for every $\mathcal{D'} \in \mathbf{P}$ such that $n_{\text{min}} \leq |D'| \leq n_{\text{max}}$, we have $\|\operatorname{dv}_r(\mathcal{D}) -\operatorname{dv}_r(\mathcal{D'})\|_1 > f + \mu$ as required.
\end{proof}

As mentioned in the introduction, Alon~\cite[Proposition 19.10]{lovasz2012large} proved that on bounded degree graphs, 
for any graph $\mathcal{G}$, radius $r$ and $\epsilon>0$ there always exists a graph $\mathcal{H}$ whose size is 
independent of $|V(\mathcal{G})|$ and whose $r$-neighbourhood distribution vector satisfies 
$\|\operatorname{dv}_r(\mathcal{G})-\operatorname{dv}_r(\mathcal{H})\|_1\leq \epsilon$.
However, the proof is only existential and does not provide an explicit bound on the 
size of $\mathcal{H}$.
As a corollary to the proof of Theorem~\ref{thrm: small dbs}, we immediately obtain explicit bounds for classes of graphs and relational databases of bounded degree whose histogram vectors form a semilinear set.

\begin{corollary}\label{cor:alon}
Let $\epsilon \in (0,1]$, $r \in \mathbb{N}$ and $\mathcal{D}$ be a $\sigma$-db that belongs to a class of $\sigma$-dbs $\mathbf{C}$ such that the set $\h_r(\mathbf{C})$ is semilinear, 
	i.e. $\operatorname{h}_r(\mathbf{C}) = M_1 \cup M_2 \cup \dots \cup M_m$ where $m \in \mathbb{N}$ and for each $i \in [m]$, $M_i = \{ \bar{v}_0^i + a_1 \bar{v}_1^i + \dots + a_{k_i} \bar{v}_{k_i}^i \mid a_1,\dots,a_{k_i} \in \mathbb{N} \}$ 
	is a linear set where $\bar{v}_0^i,\dots,\bar{v}_{k_i}^i \in \mathbb{N}^{\operatorname{c}(r)}$. 
	Then there exists a $\sigma$-db $\mathcal{D}_0$ such that 
	\[\|\operatorname{dv}_r(\mathcal{D}) -\operatorname{dv}_r(\mathcal{D}_0)\|_1 \leq \epsilon \text{ and } |D_0| \leq kv\Big(\frac{3ckv}{\epsilon} + 2 \Big)\] 
	where $c:=\operatorname{c}(r)$,  $k:=\max_{i \in [m]}k_i + 1$ and $v:=\max_{i \in [m]} \Big(\max_{j \in [0,k_i]}\|\bar{v}_j^i\|_1\Big)$.
\end{corollary}

Our aim is to construct constant time testers for hyperfinite properties whose set of $r$-histograms are semilinear. If we can approximate the $r$-neighbourhood distribution of a $\sigma$-db then by Theorem \ref{thrm: small dbs} we only need to check whether this distribution is close or not to the $r$-neighbourhood distribution of some small constant size $\sigma$-db. We let $\operatorname{EstimateFrequencies}_{r,s}$ be the algorithm that, given oracle access to an input $\sigma$-db $\mathcal{D}$, samples $s$ many elements uniformly and independently from $D$ and computes their $r$-type. The algorithm then returns the $r$-neighbourhood distribution vector of the sample.

\begin{lemma}[\cite{adler2018property}]\label{lemma: neighbourhood distribution}
Let $\mathcal{D} \in \mathbf{C}_d$ be a $\sigma$-db on $n$ elements, $\mu\in (0,1)$ and $r \in \mathbb{N}$. If $s \geq \operatorname{c}(r)^2/\mu^2 \cdot \operatorname{ln}(20\operatorname{c}(r))$, with probability at least ${9}/{10}$ the vector $\bar{v}$ returned by the algorithm $\operatorname{EstimateFrequencies}_{r,s}$ on input $\mathcal{D}$ satisfies $\| \bar{v} - \operatorname{dv}_{r}(\mathcal{D}) \|_1 \leq \mu$.
\end{lemma} 

\begin{theorem}\label{thm: constant time tester}
Let $\mathbf{C}$ be closed under removing tuples and let $\mathbf{P} \subseteq \mathbf{C}$ be a property that is hyperfinite on $\mathbf{C}$. If for each $r \in \mathbb{N}$ the set $\operatorname{h}_r(\mathbf{P})$ is semilinear, then $\mathbf{P}$ is uniformly testable on $\mathbf{C}$ in constant time in the \BDRDnew model.
\end{theorem}
\begin{proof}
Let $\epsilon \in (0,1]$. Let $r:=r(\epsilon)$ be as in Lemma \ref{lemma: locality}, let $n_{\text{min}}:=n_{\text{min}}(\epsilon)$, $n_{\text{max}}:=n_{\text{max}}(\epsilon)$, $f:=f(\epsilon)$ and $\mu:=\mu(\epsilon)$ be as in Theorem \ref{thrm: small dbs} and let $s = \operatorname{c}(r)^2/\mu^2 \cdot \operatorname{ln}(20\operatorname{c}(r))$. Assume that the set $\operatorname{h}_r(\mathbf{P})$ is semilinear. Given oracle access to a $\sigma$-db $\mathcal{D} \in \mathbf{C}$ and $|D|=n$ as an input, the $\epsilon$-tester proceeds as follows:
\begin{enumerate}
\item If $n \leq n_{\text{max}}$, do a full check of $\mathcal{D}$ and decide if $\mathcal{D} \in \mathbf{P}$. 
\item Run $\operatorname{EstimateFrequencies}_{r,s}$ and let $\bar{v}$ be the resulting vector.
\item If there exists a $\mathcal{D'} \in \mathbf{P}$ where $n_{\text{min}} \leq |D'| \leq n_{\text{max}}$ and $\|\bar{v} -\operatorname{dv}_r(\mathcal{D'})\|_1 \leq f$ then accept otherwise reject.
\end{enumerate}
The running time and query complexity of the above tester is constant as $n_{\text{max}}$ is a constant (it only depends on $\mathbf{P}$, $d$ and $\epsilon$) and $\operatorname{EstimateFrequencies}_{r,s}$ runs in constant time and makes a constant number of queries.

For correctness, first assume $\mathcal{D} \in \mathbf{P}$. By Theorem \ref{thrm: small dbs} there exists a $\sigma$-db $\mathcal{D'} \in \mathbf{P}$ such that $n_{\text{min}} \leq |D'| \leq n_{\text{max}}$ and $\|\operatorname{dv}_r(\mathcal{D})-\operatorname{dv}_r(\mathcal{D'})\|_1 \leq f - \mu$. By Lemma \ref{lemma: neighbourhood distribution} with probability at least $9/10$, $\|\bar{v} -\operatorname{dv}_r(\mathcal{D})\|_1 \leq \mu$ and therefore $\|\bar{v} -\operatorname{dv}_r(\mathcal{D'})\|_1 \leq f$. Hence with probability at least $9/10$ the tester will accept.

Now assume $\mathcal{D}$ is $\epsilon$-far from $\mathbf{P}$. By Theorem \ref{thrm: small dbs} for every $\mathcal{D'} \in \mathbf{P}$ with $n_{\text{min}} \leq |D'| \leq n_{\text{max}}$, we have $\|\operatorname{dv}_r(\mathcal{D}) -\operatorname{dv}_r(\mathcal{D'})\|_1 > f + \mu$. By Lemma \ref{lemma: neighbourhood distribution} with probability at least $9/10$, $\|\bar{v} -\operatorname{dv}_r(\mathcal{D})\|_1 \leq \mu$ and therefore for every $\mathcal{D'} \in \mathbf{P}$ with $n_{\text{min}} \leq |D'| \leq n_{\text{max}}$, $\|\bar{v} -\operatorname{dv}_r(\mathcal{D'})\|_1 > f$. Hence with probability at least $9/10$ the tester will reject.
\end{proof}

Combining Theorem \ref{thm: constant time tester} and Lemma \ref{lemma:CMSO semilinear} and the fact that $\mathbf{C}_d^t$ is hyperfinite~\cite{HassidimKNO09,AlonST90} (and so any property is hyperfinite on $\mathbf{C}_d^t$) we obtain the following as a corollary.

\begin{theorem}\label{thm: CMSO testability}
Every property $\mathbf{P}$ definable by a CMSO sentence on $\mathbf{C}_d^t$ is uniformly testable on $\mathbf{C}_d^t$ with constant time complexity in the \BDRDnew model.
\end{theorem}

\bibliography{constant_time_tester}
\newpage
\appendix
\section{Proofs of Section \ref{sec: the model}}
\begin{proof}[Proof of Lemma \ref{lemma: comparing models 1}]
Let $\pi$ be an $\epsilon$-tester, that runs in time $f(n)$, for $\mathbf{P}$ on $\mathbf{C}$ in the \BDRD model. We claim that $\pi$ is also an $\epsilon$-tester for $\mathbf{P}$ on $\mathbf{C}$ in the \BDRDnew model. Let $\mathcal{D} \in \mathbf{C}$ be the input $\sigma$-db. If $\mathcal{D} \in \mathbf{P}$ then $\pi$ will accept with probability at least $2/3$. If $\mathcal{D}$ is $\epsilon$-far from $\mathbf{P}$ in the \BDRDnew model then it must also be $\epsilon$-far from $\mathbf{P}$ in the \BDRD model and therefore $\pi$ will reject with probability at least $2/3$. Hence $\pi$ is an $\epsilon$-tester for $\mathbf{P}$ on $\mathbf{C}$ in the \BDRDnew model.
\end{proof}

\section{Proofs of Section \ref{sec: main results}}
\begin{claimproof}[Proof of Claim~\ref{claim: lower bound on size}]
By the choice of $a_j^{\mathcal{D'}}$ for $j \in [k_i]$,
\begin{align*}
&|D'| = \|\bar{v}_0^i\|_1+ \sum_{j \in [k_i]} a_j^{\mathcal{D'}} \|\bar{v}_j^i\|_1  
\geq \|\bar{v}_0^i\|_1+ \sum_{j \in [k_i]}\Big(\frac{a_j^{\mathcal{D}}n_0}{n} - \frac{1}{2}\Big) \|\bar{v}_j^i\|_1 \\
&= \|\bar{v}_0^i\|_1 - \frac{1}{2} \sum_{j \in [k_i]}\|\bar{v}_j^i\|_1 + \frac{n_0}{n}  \sum_{j \in [k_i]}a_j^{\mathcal{D}}\|\bar{v}_j^i\|_1
= \|\bar{v}_0^i\|_1 - \frac{1}{2} \sum_{j \in [k_i]}\|\bar{v}_j^i\|_1 + n_0-\frac{n_0\|\bar{v}_0^i\|_1}{n} \\
&\geq \|\bar{v}_0^i\|_1 - \frac{1}{2} \sum_{j \in [k_i]}\|\bar{v}_j^i\|_1 + n_0-\|\bar{v}_0^i\|_1
\geq -kv +  n_0 = n_{\text{min}},
\end{align*}
as $\sum_{j \in [k_i]}a_j^{\mathcal{D}}\|\bar{v}_j^i\|_1 = n - \|\bar{v}_0^i\|_1$ and $n > n_{\text{max}} \geq n_0$.
\end{claimproof}

\begin{claimproof}[Proof of Claim \ref{claim:upper bound on size}]
By the choice of $a_j^{\mathcal{D'}}$ for $j \in [k_i]$,
\begin{align*}
&|D'| = \|\bar{v}_0^i\|_1+  \sum_{j \in [k_i]}a_j^{\mathcal{D'}} \|\bar{v}_j^i\|_1
\leq \|\bar{v}_0^i\|_1+  \sum_{j \in [k_i]}\Big(\frac{a_j^{\mathcal{D}}n_0}{n} + \frac{1}{2}\Big) \|\bar{v}_j^i\|_1 \\
&= \|\bar{v}_0^i\|_1 + \frac{1}{2} \sum_{j \in [k_i]}\|\bar{v}_j^i\|_1 + n_0\Big(1-\frac{\|\bar{v}_0^i\|_1}{n}\Big) 
\leq   \sum_{0 \leq j \leq k_i}\|\bar{v}_j^i\|_1 + n_0 
\leq kv +  n_0 = n_{\text{max}},
\end{align*}
as $\sum_{j \in [k_i]}a_j^{\mathcal{D}}\|\bar{v}_j^i\|_1 = n - \|\bar{v}_0^i\|_1$.
\end{claimproof}

\end{document}

%% file: macros.tex
\newcommand{\BDRD}{$\operatorname{BDRD}$ }
\newcommand{\BDRDnew}{$\operatorname{BDRD}_{+/-}$ }
\newcommand{\h}{\operatorname{h}}